\documentclass[11pt,a4paper]{article}

\usepackage{fullpage}
\usepackage{amsthm}
\usepackage{amsmath}
\usepackage{amssymb}
\usepackage{amsthm}
\usepackage{xspace}
\usepackage[boxed]{algorithm2e}
\usepackage{mathrsfs}
\usepackage[text={16cm,23.7cm}]{geometry}

\newtheorem{definition}{Definition}
\newtheorem{proposition}{Proposition}
\newtheorem{theorem}{Theorem}
\newtheorem{corollary}{Corollary}
\newtheorem{lemma}{Lemma}
\newtheorem{property}{Property}

\newcommand{\Broder}{Broder \etal}
\newcommand{\etal}{\textsl{et al.}\xspace}
\renewcommand{\O}{\ensuremath \widetilde O}
\newcommand{\Tt}{\ensuremath \widetilde \Theta}
\newcommand{\To}{\ensuremath \widetilde \Omega}
\newcommand{\s}{\ensuremath  \mathscr S}
\renewcommand{\t}{\ensuremath \mathscr  T}
\newcommand{\E}{\ensuremath \mathbb E}
\renewcommand{\Pr}{\ensuremath \mathrm {Pr}}
\renewcommand{\*}{\hspace*{5mm}}

\author{Adrian Kosowski%
\footnote{Inria Bordeaux Sud-Ouest, 33400 Talence, France. E-mail: adrian.kosowski@inria.fr}}

\title{Faster Walks in Graphs:\\ A $\widetilde O(n^2)$ Time-Space Trade-off for Undirected $s$-$t$ Connectivity}

\date{}

\begin{document}
\maketitle\thispagestyle{empty}
\setcounter{page}{0}
\vspace{8mm}

\begin{abstract}
\noindent
In this paper, we make use of the Metropolis-type walks due to Nonaka \etal\ (2010) to provide a faster solution to the $\s$-$\t$-connectivity problem in undirected graphs (USTCON).

As our main result, we propose a family of randomized algorithms for USTCON which achieves a time-space product of $S\cdot T = \O(n^2)$ in graphs with $n$ nodes and $m$ edges (where the $\O$-notation disregards poly-logarithmic terms). This improves the previously best trade-off of $\O(n m)$, due to Feige (1995). Our algorithm consists in deploying several short Metropolis-type walks, starting from landmark nodes distributed using the scheme of Broder \etal (1994) on a modified input graph. In particular, we obtain an algorithm running in time $\O(n+m)$ which is, in general, more space-efficient than both BFS and DFS. 

We close the paper by showing how to fine-tune the Metropolis-type walk so as to match the performance parameters (e.g., average hitting time) of the unbiased random walk for any graph, while preserving a worst-case bound of $\O(n^2)$ on cover time.
\end{abstract}
\vspace{8mm}
\newpage

\section{Introduction}

In the \emph{undirected $\s$-$\t$ connectivity problem} (USTCON), the input to the algorithm is an undirected graph $G=(V,E)$ with $n$ vertices and $m$ edges. Two of the vertices of the graph, $\s, \t \in V$, are distinguished. The goal is to determine whether $\s$ and $\t$ belong to the same connected component of $G$. USTCON has a spectrum of applications in various areas of computer science, ranging from tasks of network discovery to computer-aided verification. The problem has also made its mark on complexity theory, most famously, playing a central part in the rise and eventual collapse of the complexity class SL.

The time complexity of algorithms for USTCON depends on the amount of space available to the algorithm. Given $\Tt(n)$ space, USTCON can be solved deterministically in time $O(m)$ by fast algorithms such as BFS or DFS. Given $\Theta(\log n)$ space, the problem can still be solved deterministically~\cite{Rei} in polynomial time. However, in this case the fastest known solutions are randomized. Aleliunas \etal~\cite{AKLLR} proposed a log-space algorithm with bounded error probability, which consists in running a random walk, starting from node $\s$ for $O(nm)$ steps, and testing if node $\t$ has been reached. 

The study of the interplay between the space complexity $S$ and the time complexity $T$ of randomized algorithms for USTCON was initiated by \Broder~\cite{BKRU}. They observed that both BFS/DFS, and the random walk, admit the same time-space trade-off of $T = \O (\frac{mn}{S})$, and investigated whether there exist algorithms with such a trade-off for an arbitrary choice of $S$, $c\cdot \log n \leq S \leq n$, where $c>0$ is some model-dependent constant. After a sequence of papers relying on the deployment of many short random walks, this question was eventually settled in the affirmative by Feige~\cite{F}, who proposed a family of algorithms which achieve such a time-space trade-off in the whole of the considered range of space bounds. 

The main result of this paper is an improved time-space trade-off for USTCON. Since the cover time of the random walk is precisely $\Theta(nm)$ for some graphs, any improvement with respect to Aleliunas \etal~\cite{AKLLR} or Feige~\cite{F} requires a refinement of the performed walk on graphs. Instead of the random walk, we make use of the Metropolis-Hastings walk on graphs, with weighting proposed by Nonaka~\etal~\cite{NOSY}. This walk covers any undirected graph in $\O(n^2)$ steps, but its transition probabilities rely on knowledge of the degrees of neighboring nodes at every step. We start the technical sections of this paper with an explicit implementation of the walk from~\cite{NOSY} using the Metropolis sampling algorithm from~\cite{M51}. This yields a solution to USTCON in $\O(n^2)$ time and logarithmic space. Our contribution lies in completing this quadratic time-space trade-off for larger bounds on the space complexity of the algorithm. The main technical difficulty concerns overcoming problems with short runs of the Metropolis-Hastings walk, which sometimes exhibits inferior behavior to the random walk in terms of the speed of discovering new nodes. 

For the entire range of space bounds ($c\cdot \log n \leq S \leq n$), we propose algorithms running in time $T = \O (\max\{\frac{n^2}{S}, m\})$. In other words, we obtain $T = \O (\frac{n^2}{S})$ for $S \leq \frac{n^2}{m}$, and $T = \O(m)$ for $S > \frac{n^2}{m}$. (Note that $T = \Omega(m)$ is a lower bound on execution time for any algorithm for USTCON, regardless of the space bound.) In particular, we prove that USTCON can be solved in time $\O(m)$ using space $O(\frac{n^2}{m})$, which is, in general, less than the space requirement of BFS/DFS. 

All of the considered algorithms for USTCON are randomized (in the Monte Carlo sense), with bounded probability of one-sided error. This means that the positive answer ``connected'' may only be reached by the algorithm when $\s$ and $\t$ belong the same connected component of $G$, whereas the negative answer ``not connected'' signifies that,  with probability at least $2/3$, $\s$ and $\t$ belong to different components of $G$.

\subsection{Related work}

Much of the work on undirected $\s$-$\t$ connectivity has focused around its role as the fundamental complete problem for the symmetric log-space complexity class (SL). A survey of other important problems identified as belonging to the class SL, such as simulating symmetric Turing machines, and testing if a graph is bipartite, is provided in~{AGCR}. A major line of study concerned determining the minimum space required to solve USTCON deterministically. The bound on the required space was reduced, over several decades, from the $O(\log^2 n)$ bound given by Savitch's theorem~\cite{Sav}, through $O(\log^{3/2} n)$ \cite{SZ}, and $O(\log^{4/3} n)$~\cite{ATWZ}. Finally, in 2004, Reingold's~\cite{Rei} new construction of universal graph exploration sequences provided the first log-space algorithm for USTCON, showing that SL=L. Befor Reingold's paper, Nisan~\cite{Nis} had shown a deterministic algorithm for USTCON running in polynomial time and $O(\log^2 n)$ space. Borodin \etal~\cite{BCDRT} proposed a log-space Las-Vegas type algorithm for USTCON (with no-error) which runs in expected polynomial time.

When considering randomized algorithms with bounded one-sided error, the unbiased random walk was shown to solve USTCON in $O(\log n)$ space and $\O(mn)$ time by Aleliunas \etal~\cite{AKLLR}. Several years later, \Broder~\cite{BKRU} proposed a family of algorithms based on short random walks starting from landmark nodes. Relying on landmarks chosen on the set of nodes according to the stationary distribution of the walk, they achieved a time-space trade-off of $T= \O (\frac{m^2}{S})$. Subsequent algorithms from the literature~\cite{BF,F} make use of different landmark distribution schemes. Barnes and Feige~\cite{BF} achieve a trade-off of $T = \O (\frac{m^{1.5}n^{0.5}}{S})$ by using a mixed landmark distribution scheme, which places half of the landmarks according to the stationary distribution of the random walk, and half according to the uniform distribution on nodes. Feige~\cite{F} introduces the inverse distribution scheme, which likewise places half of the landmarks according to the stationary distribution of the random walk, and the other half according to the inverse of node degrees. He achieves a time-space trade-off of $T = \O (\frac{mn /\delta}{S})$ in general, where $\delta$ is the minimum degree of the graph. Thus, the trade-off of $T = \O (\frac{n^2}{S})$ is reached for the case of (nearly) regular graphs.

Undirected $\s$-$\t$ connectivity is a special case of the more general reachability problem in directed graphs (STCON), which is a complete problem for the class NL. STCON can also be solved deterministically in $O(\log^2 n)$ space using Savitch's theorem~\cite{Sav}. So far, it has resisted fast solutions in small space. This problem was extensively studied in different variants of a model of computation based on Jumping Automata on Graphs (JAG-s). The memory of a JAG is organized in the form of $P$ pebbles placed in the graph and $Q$ states of the automaton, with space defined as $S = P \log n + \log Q$. Cook and Rackoff~\cite{CR} show a way of solving STCON in the JAG model deterministically in $O(\log^2 n)$ space, and also prove an almost matching lower bound on space of $\Omega(\log^2 n /\log \log n)$. This lower bound is also known to apply to randomized JAG-s running in slightly super-polynomial time~\cite{BS}. Gopalan \etal~\cite{GLM} propose a family of algorithms for STCON based on short random walks, whose runtime increases from $O(n^{\log n})$ to $O(n^n)$ as space decreases from $O(\log^2 n)$ to $O(\log n)$.

Finally, we remark on recent developments in the area of graph exploration with biased random walks. Ikeda \etal~\cite{IKY} and Nonaka \etal~\cite{NOSY} studied possible adjustments to the transition matrix of the walk based on the availability of local topological information (otherwise known as ``look-ahead''). In general, the idea of these approaches is to increase the probability of transition to a node of lower degree. The former paper introduces a new type of walk, called the $\beta$-walk, whose transition matrices are biased so that transition from a node to its neighbor of degree $d$ is proportional to $d^{-\beta}$. Such a walk was shown to visit all nodes of the graph in $O(n^2 \log n)$ steps in expectation for an optimal choice of parameter $\beta = 1/2$. Nonaka \etal~\cite{NOSY} later used the key lemmas from this work to prove an analogous result for a walk with a modified transition matrix, which fits into the class of Metropolis-Hastings walks. This walk is the starting point for considerations in our paper. A somewhat different approach was adopted by Berenbrink \etal~\cite{BCERS}, who show that a random walk with the additional capability of marking one unvisited node in its neighborhood as visited can be used to speed up exploration.

\subsection{Overview of the paper}

The organization of the technical parts of the paper is the following. In Section~\ref{sec2}, we recall the definition of the Metropolis-Hastings walk and provide its efficient implementation using the Metropolis algorithm. In this way, given a representation of graph $G$, each step of the walk can be simulated by a procedure running in $\O(1)$ time and using $\Theta (\log n)$ bits of space.

We subsequently identify the key properties of the unit-potential Metropolis-Hastings walk, denoted $RW(G_1)$, which allow it to be used as a replacement for the (unbiased) random walk on $G$, denoted $RW(G)$, in algorithms solving USTCON. The major difference between these types of walks is that a short random walk $RW(G)$ has the desirable property of \emph{low edge-return rate}, i.e., each edge of the graph is visited $O(\sqrt t)$ times in expectation during $t$ steps of the walk (for sufficiently small $t$). However, no analogous property hold for the Metropolis-Hastings walk. In fact, on some graphs (e.g.,\ the \emph{glitter star} defined in~\cite{NOSY}), the Metropolis-Hastings walk $RW(G_1)$, will in expectation discover only $O(1)$ edges during $t$ steps of the walk, visiting each of these edges $\Omega(t)$ times (for any choice of $t\leq n$). We overcome this problem in two stages:
\begin{itemize}

\item In Section~\ref{sec22} we prove that in a graph of maximum degree $\Delta$, the Metropolis-Hastings walk $RW(G_1)$ begins to achieve a \emph{low node-return rate} starting from a threshold length of $\Delta^2$ steps: a Metropolis-Hastings walk of length $t$, $\Delta^2 < t < n^2$, visits each node of the graph $O(\sqrt t)$ times in expectation. This property is formally stated as Lemma~\ref{lemiii}. 

\item In Section~\ref{sec3} we show how to obtain the trade-off $T = \O (\max\{\frac{n^2}{S}, m\})$  for an arbitrary choice of space bound $S$. Our initial approach makes use of a modification of a technique introduced by \Broder~\cite{BKRU}. It consists in running $p \approx S$ walks of length $t \approx \frac{n^2}{S^2}$ each, which originate from an appropriately chosen subset of $p$ nodes of the graph called \emph{landmarks}. In our formulation, the walks used are Metropolis-Hastings walks (rather than random walks on $G$), and the set of landmarks is sampled uniformly on $V$. By observing the visits of each of these walks to other landmarks from the set, it is possible to obtain information about paths connecting different landmarks. When the performed Metropolis-Hastings walks have a low node-return rate (i.e., when $t>\Delta^2$), the obtained information turns out to be w.h.p.\ sufficient to find an answer to $\s$-$\t$ connectivity with a low probability of error. Otherwise, when $t < \Delta^2$, we modify the approach, performing a logical transformation of graph $G$. We split each node of degree greater than $\sqrt t$, so that the maximum degree of the modified graph does not exceed $\sqrt t$. Then, all of the considerations are performed for this modified graph. In particular, the set of landmark nodes is chosen by uniform sampling on the set of nodes of this modified graph. The overhead associated with this transformation is just small enough for our algorithm to have the claimed time complexity of $T = \O (\max\{\frac{n^2}{S}, m\})$. An implementation of the complete algorithm is provided in Appendix~A.
\end{itemize}

Finally, in the closing Section~\ref{sec5} we discuss the tightness of the obtained results. We also propose a modified weighting of the Metropolis-Hastings walk which performs faster than uniform-weighted Metropolis-Hastings for many classes of graphs, while still covering all the nodes of the graph in $\O(n^2)$ time. This walk satisfies the property that its commute time between any pair of nodes (and consequently also the average hitting time) is asymptotically upper-bounded by the values of the respective parameters for the unbiased random walk. In particular, it covers all the nodes of the previously mentioned glitter star, in expected $\O(n)$ steps.

\subsection{Notation and model}

The input graph $G=(V,E)$, with $|V|=n$ and $|E|=m$, is simple and not necessarily connected. In order to simplify notation for complexity bounds, we assume $m = \Omega(n)$. The degree of a node $v\in V$ is denoted by $\deg(v)$, the neighborhood of node $v$ by $\Gamma(v)$, and  the closed neighborhood of $v$ by $\Gamma^+(v) = \Gamma(v) \cup \{v\}$. The maximum degree of the graph is denoted by $\Delta$. The arc set $\vec E \subseteq V \times V$ of undirected graph $G$ is understood as the set of arcs of all edges and self-loops of $G$: $\vec E = \{(v,u) : v\in V, u \in \Gamma^+(v)\}$. An arc $(v,u)\in \vec E$ is sometimes denoted as $e_{vu}$ for compactness of notation. Note that the symbols $V$, $E$, $\Delta$, $n$, $m$ always refer to the input graph $G$. When considering a different graph $X$, we will sometimes denote its vertex, edge, and arc sets by $V(X)$, $E(X)$, and $\vec E(X)$, respectively. 

Our algorithms are designed for the classical RAM model of computation. No special assumptions are made on the representation of graph $G$, except that for any node $v\in V$, there should exist a local ordering on the set of its neighbors, given by the bijective function $PORT_v : \Gamma(v) \to \{0,1,\ldots,\deg(v)-1\}$. Each of the following operations should be possible to implement in $\O(1)$ time: computing $\deg(v)$, computing $PORT_v(u)$ for a node $u \in \Gamma(v)$, and ``traversing an edge'' by computing $PORT^{-1}_v(i)$, for port $i\in \{0,1,\ldots,\deg(v)-1\}$. An example of a permissible representation is a lexicographically sorted array of ordered pairs of identifiers of neighboring nodes $(u,v)$, taken over $\{u,v\}\in E$.

For most of the paper, we consider weighted reversible Markovian processes corresponding to a random walk $RW(X)$ on some weighted undirected graph $X$ with positive weights on arcs. The walk is located on the nodes of graph $X$, and the next state of the walk is reached by following an arc incident to the current node, chosen with probability proportional to the weight of this arc. By a slight abuse of notation, we denote the transition matrix of the walk in the same way as the weighted graph. Most other notation follows that of Aldous and Fill~\cite{AF}. In particular, we consider the following random variables:
\begin{itemize}
\item $N_a (t)$ denotes the number of steps in the time interval $[0,t)$ during which the walk visits $a$, where the symbol $a$ may represent a node, edge, or arc of the graph.
\item $T_a$ denotes the first moment of time $t>0$ at which the walk first visits (or returns to) a node from $a$, where the symbol $a$ may represent a subset of nodes or a single node of the graph.
\end{itemize}
By writing $\E_\alpha Y$ and $\Pr_\alpha [E]$, respectively, we mean the expectation of random variable $Y$, and the probability of event $E$ occurring, taken over walks starting from probability distribution $\alpha$ (which may be concentrated on a single node or arc). A walk starting from an arc is understood as one which starts from the head of the arc at time $0$, and then moves to the tail of the arc at time $1$.

Given a weighted graph $X$, we denote by $Com(i,j) \equiv \E_i T_j + \E_j T_i$ the commute time between nodes $i,j \in V(X)$. Throughout the paper, we consider only walks representing reversible Markovian processes, corresponding to symmetric weightings of the graph: $w(e_{vu}) = w(e_{uv})$, for all $(u,v)\in \vec E$. In some of the proofs, we rely on the resistor network representation of reversible walks: for each edge $e=\{u,v\} \in E(X)$ having weight $w(e)$ on each of its arc, a resistor with resistance $1/w(e)$ is placed between nodes $u$ and $v$ of the resistor network. The symbol $R(u,v)$ denotes the resistance of replacement between nodes $u$ and $v$ of the network. We recall that $Com(i,j) = R(i,j)\sum_{e\in \vec E(X)}w(e)$.~\cite{CRRST}

\section{Preliminaries: The Metropolis-Hastings Walk on Graphs}\label{sec2}\label{sec22}

The \emph{Metropolis-Hastings walk} with potential function $f : V \to \mathbb{R}^+$ is defined as a walk on the weighted graph $G_f = (V, E, w_f)$ with the following assignment of weights $w_f : \vec E \to \mathbb{R}^+$:
$$
	w_f(e_{vu}) = \min\left\{\frac{f(v)}{\deg(v)}, \frac{f(u)}{\deg(u)}\right\}, \textrm{\quad for all } \{v,u\}\in E.
$$
$$
	w_f(e_{vv}) = f(v) - \sum_{u\in \Gamma(v)} w_f(\{v,u\}), \textrm{\quad for all } v \in V,
$$
We recall that for a walk in state $v \in V$, the next state is chosen as $u \in \Gamma(v) \cup \{v\}$ with probability proportional to the weight $w_f(e_{vu})$. By a classical result due to Metropolis \etal~\cite{M51}, for a given representation of graph $G$, a single step of the Metropolis-Hastings walk $RW(G_f)$ can be simulated in $\O(1)$ time and space by means of the procedure shown in Algorithm~\ref{algo1}. The algorithm takes advantage of the fact that $w_f(e_{vu}) / \sum_{x\in \Gamma^+(v)} w_f(e_{vx})\leq \frac{1}{\deg (v)}$, for all $u \in \Gamma(v)$. For a walk located at node $v$, it samples a node $u\in \Gamma(v)$ with uniform probability $\frac{1}{\deg (v)}$, and accepts $u$ as the new state with the appropriate probability. We remark that a step of $RW(G_f)$ can  also be simulated by a log-space automaton which pushes a pebble along the arc $(v,u)$. The pebble remains at $u$ if state $u$ is accepted, and otherwise reverts to $v$ by traversing the arc $(u,v)$. Thus, one step of $RW(G_f)$ can be simulated by at most two moves of a pebble.

\IncMargin{0.5em}
\begin{algorithm}
\textbf{function} next\_state ($v$: node) \{\\
\* $u \gets$ neighbor of $v$ in $G$ chosen uniformly at random; \quad// \emph{pick a new state}\\
\* \textbf{with probability} $\min \{\frac{\deg(v)}{\deg(u)}\frac{f(u)}{f(v)}, 1\}$ \textbf{do} \textbf{return} $u$; \ // \emph{accept: move to new state}\\
\* \textbf{return} $v$; \quad// \emph{do not accept: keep current state}\\
\}
\caption{State transition function on $V$ for the walk $RW(G_f)$.
}
\label{algo1}
\end{algorithm}

\begin{definition}
We denote by $G_1$ the weighted graph $G_f$ for the unit potential function $f(v) \equiv 1$.
\end{definition}

From the next two sections, we focus on the Metropolis-Hastings walk $RW(G_1)$. We note that the weights on the edges of $G$ are now simply given by $w(e_{vu}) = \min\{\frac{1}{\deg(v)}, \frac{1}{\deg(u)}\}$.

The bound on the time required by the Metropolis-Hastings walk to discover w.h.p.\ the entire connected component containing the starting node of the walk follows from the considerations of Nonaka \etal~\cite{NOSY}. (All omitted proofs are provided in the Appendix.)

\begin{lemma}[\cite{NOSY}]\label{lemtwo}
Let $i \in V$, let $H$ be the connected component of $G$ containing node $i$, and let $n_H = |V(H)|$. Then:
\begin{itemize}
\item a walk $RW(G_1)$ of length $12 n_H^2$ starting from $i$ covers an arbitrary node $j\in V(H)$ with probability at least $\frac{1}{2}$.
\item a walk $RW(G_1)$ of length $24 n_H^2 \log n$ starting from $i$ covers all nodes from $V(H)$ with probability at least $1 - \frac{1}{n}$.
\end{itemize}
\end{lemma}

By the above Lemma, a solution to USTCON, with probability $1-\frac{1}{n}$, is obtained by running the walk $RW(G_1)$, starting from $\s$, for $24 n^2 \log n$ steps. USTCON can therefore be solved in log-space by running Algorithm~\ref{algo1} in a loop for $O(n^2 \log n)$ iterations. (We are unaware of any previous reference in the literature for this observation.)

\begin{corollary}
There is a log-space algorithm for USTCON which runs in time $O(n^2\log n)$, with probability of one-sided error bounded by $\frac{1}{n}$.
\qed
\end{corollary}

For our purposes, we will need a more detailed analysis of the behavior of the Metropolis-Hastings walk. We start by recalling that the Metropolis-Hastings walk $RW(G_1)$ is a reversible Markovian process, since $w(e_{vu}) = w(e_{uv})$ for all arcs. Its stationary distribution is the uniform distribution $\pi : V \to \mathbb{R}^+$, with
$
\pi(v)= \frac{1}{n},
$
for all $v\in V$. This allows us to show the following key lemma which captures the ``low node-return rate'' property of the Metropolis-Hastings walk, as highlighted in the Introduction. The first claim states that a Metropolis-Hastings walk starting within any subset of nodes $A\subsetneq V$ is likely to leave it within $O(|A|^2)$ steps, while its second claim shows that a Metropolis-Hastings walk of length $t$ is likely to return to its starting node not more than $O(\sqrt t)$ times. However, both of the above statements hold only when considering walks of duration $\Omega(\Delta^2)$.
\begin{lemma}\label{lemiii}
Suppose that $G$ is connected. Let $A \subsetneq V$, and let $i\in A$. For a weighted random walk $RW(G_1)$ starting from node $i$:
\begin{enumerate}
\item[(i)] the expected time to reach a node from $V\setminus A$ is bounded by:
$$
\E_i T_{V\setminus A} < (|A|+1)(6|A|+2\Delta),
$$
\item[(ii)] the expected number of visits to node $i$ before any time $t$, $0<t<6n^2$, is bounded by:
$$
\E_i N_i (t)  < 5 \sqrt t + 2 \Delta.
$$
\end{enumerate}
\end{lemma}
The proof of the lemma follows by an analysis of resistances of replacement along shortest paths in the resistor network for the weighted graph $G_1$.

\section{A time-space trade-off for USTCON}\label{sec3}

The time-space tradoffs for USTCON proposed by \Broder~\cite{BKRU} make use of a number of short random walks, originating from a subset of nodes of the graph called \emph{landmarks}. Herein, we design an algorithm which replaces these random walks by Metropolis-Hastings walks.

We start by a brief overview of the landmark-based approach. When considering an algorithm using space $S$, the size of the set of landmarks is defined by a parameter $p = \Theta(S/\log n)$. The algorithm first chooses a set of landmarks $L \subseteq V$ consisting of $p+2$ nodes: node $\s$, node $\t$, and $p$ nodes picked (in the case of our work) uniformly at random from $V$. Then, a walk of suitably chosen length $t$ is released from each of the landmarks. Throughout this process, the algorithm maintains a disjoint-set data structure (also known as ``Union-Find''~\cite{HU}) on the set of landmarks, with each set corresponding to the landmarks identified as belonging to the same connected component of the graph.

Initially, each landmark belongs to a separate set. Whenever a walk starting from one landmark hits some other landmark, the algorithm updates the data structure, merging the classes corresponding to these two landmarks. At the end of the process, if landmarks $\s$ and $\t$ belong to the same class, then, with certainty, there exists an $\s$-$\t$ path in $G$, and the answer to USTCON is positive. Otherwise, the algorithm returns a negative result, and, in the rest of this Section, we focus on proving that this result is correct w.h.p.

The runtime of the algorithm of \Broder~is determined by the time of running $p = \Tt(S)$ random walks of length $t$ each, thus $T = \O(tp) = \O(\frac{t p^2}{S})$. To achieve the claimed trade-off of $T = \O(\frac{n^2}{S})$, we will therefore need to use walks of length roughly $t \approx \frac{n^2}{p^2}$.

\subsection{An initial approach}\label{sec31}

We fix a value of the parameter $p = O(S)$, describing the number of landmark nodes. The landmark-based algorithms are built around the premise that landmarks belonging to the same connected component of $G$ quickly discover each other with the help of the short walks they release. In particular, it is desirable that the set of landmarks in each connected component of $G$ has the property that for any partition of the set of landmarks into two subsets, a short walk originating from a landmark in one of these subsets is likely to reach some landmark from the other subset. \Broder~\cite{BKRU} observe (cf.\ also~\cite{F} for a high-level exposition of the argument) that this property is satisfied if the considered set of landmarks is \emph{good}, i.e., it fulfills the following two assumptions. Firstly, the set of short walks originating from all of the landmarks should be likely to jointly cover all the arcs of the graph. Secondly, a short walk originating from an arbitrary starting node of the graph should be likely to reach at least one landmark from the set.

Most of the analysis and key lemmas in this subsection follow along the lines proposed by \Broder We confine ourselves to a summary of the approach, highlighting the subtle differences resulting from the use of Metropolis-Hastings walks. We start by re-setting the good landmark property in the context of Metropolis-Hastings walks $RW(G_1)$ of a specifically chosen length $\tau$.

\begin{property}\label{begood}
Let $L \subseteq V$ be the set of $p = |L|$ landmark nodes, let $H$ be a connected component of $G$, and let $n_p = \max\{\gamma \frac{n}{p} \log n, \Delta\}$, where $\gamma =60$ is a suitably chosen absolute constant (whose value follows from the proof of Lemma~\ref{lemGood}). We say that the set of landmarks $L$ is \emph{good with respect to $H$} if the following properties hold:
\begin{itemize}
\item With probability at least $1 - \frac{1}{n}$, a set of $p$ walks $RW(G_1)$ of length $\tau = n_p^2$ each, with one walk originating from each landmark from $L$, covers an arbitrarily chosen arc of $H$.
\item With probability at least $1 - \frac{1}{n}$, a walk $RW(G_1)$ of length $\tau = n_p^2$, originating from an arbitrarily chosen node of $H$, hits some landmark from $L$.
\end{itemize}
\end{property}

In the above property, the choice of the length $\tau$ of the walk takes into account that walks of length $\O(\frac{n^2}{p^2})$ lead us to the sought time complexity of $\O(\frac{n^2}{p})$ for the algorithm. However, in order to ensure that a uniformly sampled landmark set is likely to be good, we will make use of the low node-return rate of the Metropolis-Hastings walk from Lemma~\ref{lemiii}, and thus we need to have $\tau = \Omega(\Delta^2)$.

We will now show that that Property~\ref{begood} holds w.h.p.\ for a set of landmarks, each of which is chosen according to the uniform distribution $\pi$ on the set of nodes $V$. To achieve this, we capture the ``contribution'' of a single Metropolis-Hastings walk to the probability of success of the events described in the Property. It turns out that a Metropolis-Hastings walk  of the chosen length $\tau$, when starting from a landmark, has probability $\Omega(1/p)$ of reaching an arbitrary arc of the graph. When starting from an arbitrary node from $V$, such a walk has probability $\Omega(1/p)$ of reaching any specific landmark. These claims are formulated in a slightly more general way as the two lemmas below. Their proofs take into account the low node-return rate property from Lemma~\ref{lemiii}$(ii)$, and the properties of a walk starting from its stationary distribution $\pi$.

\begin{lemma}\label{lemA}
Suppose that $G$ is connected. For a weighted walk $RW(G_1)$ starting from a node chosen according to the uniform distribution $\pi$, the probability of traversing (a fixed) non-loop arc $e_{ij}$ before time $t$, where $\Delta^2 \leq t < 6n^2$, is:
$$
\Pr_{\pi} [T_{e_{ij}} < t ] > 0.1 \sqrt t / n.
$$
\end{lemma}

\begin{lemma}\label{lemB}
Suppose that $G$ is connected. Let $j\in V$ be picked according to the uniform distribution $\pi$. For a weighted walk $RW(G_1)$ starting from some node $i \in V$, the probability of reaching $j$ before time $t$, where $\Delta^2 \leq t < 6n^2$, is:
$$
\Pr_{i} [T_j < t ] > 0.1 \sqrt t / n.
$$
\end{lemma}

After combining the above lemmas and applying some elementary arguments about unions of independent events, we finally obtain that Property~\ref{begood} is satisfied w.h.p.\ by landmarks uniformly chosen from $V$, provided that the considered connected component is sufficiently large.

\begin{lemma}\label{lemGood}
If a connected component $H \subseteq G$ has $n_H \geq n_p/6$ nodes, then a (multi)set of $p$ nodes, picked with uniform probability from $V$, is a good set of landmarks with respect to $H$ with probability at least $1- \frac{1}{2n}$.
\end{lemma}

The results of \Broder imply directly that if a set of landmarks is good with respect to a connected component $H$, then all landmarks in $H$ can be identified as belonging to the same connected component by releasing a small number of walks from each landmark, and applying Union-Find type operations on a disjoint-set datastructure on the landmarks. Since the proof does not rely on any other assumptions beyond the properties of good landmarks, the result is directly applicable to our considerations of the Metropolis-Hastings walk.

\begin{lemma}[\cite{BKRU}]\label{bro}
Let $L$ be a set of good landmarks with respect to connected component $H \subseteq G$. Then, a set of walks of length $n_p^2$ each, with $\beta \log n$ walks originating from each of the landmarks, with probability at least $1-\frac{1}{2n}$ discovers that all landmarks located within $H$ belong to the same connected component.
\end{lemma}
\noindent
In the above, the absolute constant $\beta$ can be chosen as $\beta = 72$.

Our algorithm for USTCON is now obtained as follows. We pick a set of landmarks $L$, consisting of $\s$, $\t$, and $p$ nodes picked uniformly at random from $V$, and then follow $\beta \log n$ Metropolis-Hastings walks from each landmark, updating the disjoint-set data structure. Finally, the algorithm decides whether $\s$ and $\t$ are connected based on whether these two landmarks have been identified as belonging to the same connected component.

The algorithm never provides a false-positive answer. The probability of identifying a pair of nodes $\s$ and $\t$ from the same component $H\subseteq G$ as not being connected, can be bounded using the following argument adapted from \Broder Let $H$ be the connected component of $G$ containing node $\s$. If $n_H \geq n_p / 6$, then by Lemma~\ref{lemGood}, the set $L$ is a set of good landmarks with respect to $H$ with probability at least $1-\frac{1}{2n}$ (note that adding nodes $\s$ and $\t$ to a good set of landmarks cannot make this set of landmarks a bad one). Conditioned on this, by Lemma~\ref{bro}, we obtain a correct answer to USTCON with probability $1-\frac{1}{2n}$. Thus, the algorithm works correctly with probability at least $1-\frac{1}{n}$. In the case when $n_H < n_p / 6$, we consider only the walks originating from landmark $\s$. There are $\beta \log n$ such (independent) walks, each of length $n_p^2 > 36 n_H^2$. It follows from Lemma~3, putting $i = \s$ and $j = \t$, that in this case, node $\t$ will be reached with probability at least $1 - \frac{1}{n}$. This completes the proof of correctness.

\begin{proposition}
For all $p\geq 1$, there is an algorithm solving USTCON using space $S = \O(p)$ and time $T=\O(n_p^2 p)$, where $n_p = \max\{\gamma \frac{n}{p} \log n, \Delta\}$, with probability of one-sided error bounded by $\frac{1}{n}$.\qed
\end{proposition}

For the case when $p = \O(\frac{n}{\Delta})$, we have obtained the trade-off $T = \O(\frac{n^2}{S})$. We now show how to obtain the claimed trade-off in the general case.

\subsection{Removing the dependence on maximum degree}\label{sec32}

We now remove the dependence of length of the used walks on the value of $\Delta$. We design a graph $G^*=(V^*,E^*)$ by subdividing the nodes of $G$, so that each node from $V$ turns into a path of nodes in $G^*$ with maximum degree bounded by $D+2$, where $D \geq 1$ is an integer parameter, whose value is specified later. Formally, graph $G^*$ is defined as follows:
\begin{itemize}
\item For each node $v\in V$, $V^*$ contains $\lceil \frac{\deg(v)}{D}\rceil$ copies of $v$, labeled $(v,0),(v,1),\ldots, (v,\lceil \frac{\deg(v)}{D}\rceil-1)$.
\item Nodes $(u,i)$ and $(v,j)$, $u\neq v$, are connected by an edge in $E^*$ if and only if $\{u,v\}\in E$, $iD \leq PORT_v(u) < (i+1)D$, and $jD \leq PORT_u(v) < (j+1)D$.
\item Nodes $(u,i)$ and $(u,i+1)$, for all $0 \leq i <  \lceil\frac{\deg(v)}{D}\rceil-1$, are connected by an edge of $E^*$, with special port labels $`prev'$ and $`next'$ at its endpoints.
\end{itemize}
Let $n^*$ and $\Delta^*$ be the number of nodes and the maximum degree of $G^*$, respectively. We have $\Delta^* = D + 2$, and the following bound on $n^*$ holds:
$$
n^* = \sum_{v \in V} \lceil \frac{\deg(v)}{D}\rceil <  n + \sum_{v \in V} \frac{\deg(v)}{D} = n + \frac{2m}{D}.
$$
Solving USTCON on $G$ between nodes $\s$ and $\t$ can be reduced to solving USTCON on $G^*$ between nodes $(\s, 0)$ and $(\t, 0)$, since the transformation of $G$ into $G^*$ does not affect connectivity. In order to apply the algorithm for USTCON to $G^*$, rather than to $G$, we introduce the following modifications:
\begin{itemize}
\item Landmarks need to be distributed following the uniform distribution on $V^*$. This can be achieved by picking $p$ integers uniformly at random from the range $[1, n^*]$, then enumerating all the nodes of $V^*$ in order, and associating the landmarks with the corresponding nodes from $V^*$. This operation can be performed in $O(n^* + p\log n^*)$ time, which is always $\O(m)$.
\item The performed walks need to follow $RW(G^*_1)$, rather than $RW(G_1)$. A simulation of one step of the walk $RW(G^*_1)$ can be performed in $\O(1)$ time.
\item The duration of each of the performed walks is given as $n_p^{*2}$, where:
\begin{equation}\label{eqn0}
n_p^* = \max\left\{\gamma \frac{n^*}{p} \log n^*, \Delta^*\right\} < \max \left\{ \gamma\frac{n + 2m/D}{p} \cdot 2\log n, D+2\right\}
\end{equation}
\end{itemize}
It follows that the time complexity of the entire algorithm is bounded by the $\O(m)$ complexity of landmark distribution and the $\O(n_p^{*2} p)$ complexity of simulating the Metropolis-Hastings walks on $G^*$. Substituting the expression from~\eqref{eqn0} for $n_p^*$, we have:
$$
T = \O (m + n_p^{*2} p) = \O \left (m + \frac{n^2}{p} + \frac{m^2}{D^2 p} + D^2 p\right)
$$
Now, putting $D = \lceil\sqrt{m/p}\rceil$ gives $D^2 p = \Theta(m)$, and we obtain the required time bound $T = \O(m + \frac{n^2}{p} + \frac{m^2}{m} + m) = \O(\max\{\frac{n^2}{p}, m\})$. Since the proposed solution can be implemented with a space bound of $S = \O(p)$, we have proven the main theorem of the paper.

\begin{theorem}\label{thm:main}
For all $S \geq c\log n$, where $c>0$ is some model-dependent constant, there is an algorithm solving USTCON using space $S$ and time $T=\O(\max\{\frac{n^2}{S}, m\})$, with probability of one-sided error bounded by $\frac{1}{n}$.\qed
\end{theorem}

\section{Remarks}\label{sec5}

\paragraph{Tightness of the trade-off.}\label{sec51}

For a space bound $S \geq \frac{n^2}{m}$, we cannot hope for an algorithm with smaller run-time than $T=\O(m)$, achieved in Theorem~\ref{thm:main}.
In fact, the lower bound of $T = \Omega(m)$ holds for the RAM model under most reasonable representations of $G$ in the memory (cf. Proposition~\ref{lbm} in Appendix B for a standard proof).

For smaller values of $S$, the optimality of the achieved trade-off $S\cdot T = \O(n^2)$ is open. For the extremal case of $S = O(\log n)$, the results of~\cite{BBRRT} imply that $T=\To(n^2)$ for any deterministic algorithm using a jumping automaton (JAG) with at most one movable pebble. There is also little hope of improving the time complexity using randomized algorithms similar to the Metropolis-Hastings walk, since Nonaka \etal~\cite{NOSY} showed that any walk, having a stationary distribution which is (almost) uniform on the nodes of the graph, has $\Omega(n^2)$ cover time for some graphs.

Even more strongly, one can ask whether there exists an algorithm for USTCON which runs in $\O(1)$ space and $\O(m)$ time. This appears unlikely in view of the negative result of Edmonds~\cite{E}, who showed that a randomized JAG using $\O(1)$ space and $o(\log n / \log \log n)$ pebbles requires in expectation $n^{1+ \Omega(1)/\log\log n}$ time to explore certain $3$-regular graphs.

\paragraph{Fine-tuning the Metropolis-Hastings walk.}\label{sec52}

In view of Lemma~\ref{lemtwo}, the Metropolis-Hastings walk visits all the nodes of a graph within $\O(n^2)$ steps. This is an improvement with respect to the bound of $O(nm)$ on the cover time of an unbiased random walk. Nevertheless, the Metropolis-Hastings walk may perform worse than the random walk for specific graph classes. A generic example of such a graph, called the \emph{glitter star}, was defined by~\cite{NOSY} as a tree on $n = 2l+1$ nodes, with one central node of degree $l$ connected to $l$ nodes of degree $2$, which are in turn connected to $l$ leaves. On the glitter star, the cover time of the random walk is $\Theta(n \log n)$, and the cover time of the Metropolis-Hastings walk is $\Theta(n^2)$.

Below we propose a walk $RW(G_f)$ with a different potential function which combines some of the advantages of the random walk and the Metropolis-Hastings walk.

\begin{proposition}\label{probla}
For a graph $G$, let the node potential function $f : V \to \mathbb{R}^+$ be given as $f(u) = \frac{\deg(u)}{d} + 1$, where $d = \frac{2m}{n}$ is the average degree of the graph.  Then, for any pair of nodes $u,v\in V$, the walk $RW(G_f)$ achieves a commute time of:
$$
Com_{G_f}(u,v) = O(\min\{Com_{G}(u,v),Com_{G_1}(u,v)\}),
$$
where $Com_{G}(u,v)$ and $Com_{G_1}(u,v)$ denote the commute times for the random walk on $G$ and the Metropolis-Hastings walk, respectively. A step of the walk $RW(G_f)$ can be simulated using $\O(1)$ space and time.
\end{proposition}

The above Proposition implies that for any graph, the walk $RW(G_f)$ with $f(u) = \frac{\deg(u)}{d} + 1$, is asymptotically not slower than the unbiased random walk in terms of parameters such as maximum hitting time and (arbitrarily weighted) average hitting time. At the same time, this walk preserves the upper bound of $\O(n^2)$ on the cover time in the graph, making it an interesting alternative to the unbiased random walk in practical applications, e.g., for different random graph models.

We remark that there exist different ways of combining the unbiased random walk and the Metropolis-Hastings walk. For example, one may consider an automaton which iteratively performs a phase of the walk $RW(G)$, followed by a phase of the walk $RW(G_1)$ of the same length, doubling the lengths of both walks in each subsequent iteration. Such a walk visits all the nodes of the graph in expected time asymptotically equal to the cover time of the faster of the two walks.

\newpage
\pagenumbering{roman}
\setcounter{page}{1}

\newpage
\section*{Appendix A: Implementation}

For the sake of completeness, below we provide the pseudocode of the algorithm for USTCON announced in Theorem~\ref{thm:main}. The implementation is self-contained, except for the following subroutines. The disjoint-set data structure is implemented by the procedures: $SET(x)$ which adds a new set containing only element $x$ to the data structure, $FIND(x)$ which returns (the identifier of) the set containing element $x$, and $UNION(S_1, S_2)$ which replaces sets $S_1$ and $S_2$ by set $S_1\cup S_2$ in the data structure. Each of these operations is performed in amortized $\O(1)$ time.

The routine $TRAVERSE\_EDGE_v(port)$, for a node $v\in V$, returns a pair $(u, inport)$, such that $u \in \Gamma(v)$ with $PORT_v(u) = port$, and $inport = PORT_u(v)$. This routine can be performed in $\O(1)$ time in the RAM model, as well as in most JAG-based models.

We recall the values of the absolute constants: $\gamma = 60$ and $\beta = 72$.

\medskip
\noindent
// \emph{Solution to USTCON using $p$ auxiliary landmarks}\\
\textbf{procedure} test\_connectivity ($\s,\t$: nodes from $V$) \{\\
\* $D \gets \lceil \sqrt{m / p} \rceil$;\\
\* $n^* \gets \sum_{v \in V} \lceil \frac{\deg(v)}{D}\rceil$;\\
\* $L \gets \{(\s,0),(\t,0)\}$;\\[2mm]
\* // \emph{Distribute $p$ landmarks uniformly on $V^*$}\\
\* $\ell \gets$ multi-set of $p$ integers chosen uniformly at random from the range $\{1,2,\ldots,n^*\}$;\\
\* $i \gets 0$;\\
\* \textbf{for} $v\in V$ \textbf{do}\\
\* \* \textbf{for} $j \gets 0,1,\ldots, \lceil \frac{\deg(v)}{D}\rceil-1$ \textbf{do} \{\\
\* \* \* $i \gets i+1$;\\
\* \* \* \textbf{if} $i \in \ell$ \textbf{then} $L \gets L \cup \{(v, j)\}$;\\
\* \* \}\\
\* \textbf{for} $l\in L$ \textbf{do} $SET(l) \gets l$;\\[2mm]
\* // \emph{From each landmark, run $\beta \log n^*$ Metropolis walks $RW(G^{*}_1)$ of length $n_p^{*2}$ each}\\
\* \textbf{repeat} $\beta \log n^*$ times \{\\
\* \* \textbf{for} $l\in L$ \textbf{do} \{\\
\* \* \* $s \gets l$; \\
\* \* \* \textbf{repeat} $\lceil \max\{ \gamma \frac{n^*}{p} \log n^*, D+2\} \rceil^{2}$ times \{\\
\* \* \* \* $s \gets$ next\_state* ($s$);\\
\* \* \* \* $UNION (FIND(s), FIND(l))$;\\
\* \* \* \}\\
\* \* \}\\
\* \}\\
\* \textbf{if} $FIND((\s,0)) = FIND((\t,0))$ \textbf{then} \textbf{return} ``\emph{connected}'';\\
\* \textbf{return} ``\emph{probably not connected}'';\\
\}\\
\\
// \emph{Simulate one step of the walk $RW({G^{*}_1})$ from state $(v,i)\in V^*$}\\
\textbf{function} next\_state* ($v$: node, $v\_i$: integer) \{\\
\* $v\_deg^* \gets $get\_degree*$(v, v\_i)$;\\
\* $port \gets$ get\_random\_port*$(v, v\_i)$;\\
\* \textbf{if} $port = `prev'$ \textbf{then} \{\\
\* \* $u \gets v$;\\
\* \* $u\_i \gets v\_i-1$;\\
\* \} \textbf{else if} $port = `next'$ \textbf{then} \{\\
\* \* $u \gets v$;\\
\* \* $u\_i \gets v\_i+1$;\\
\* \} \textbf{else} \{ \emph{// $port \in [left, right]$ is an integer corresponding to a port at $v$ in $G$}\\
\* \* $(u, inport) \gets TRAVERSE\_EDGE_{v} (port)$;\\
\* \* $u\_i \gets \lfloor \frac{inport}{D}\rfloor$;\\
\* \}\\
\* $u\_deg^* \gets$ get\_degree*$(u, u\_i)$;\\
\* \textbf{with probability} 
$\min \{\frac{v\_deg^*}{u\_deg^*}, 1\}$ \textbf{do} \textbf{return} $(u, u\_i)$;\\
\* \textbf{return} $(v, v\_i)$;\\
\}\\
\\
// \emph{Return the degree of $(v,i)$ in $G^*$}\\
\textbf{function} get\_degree* ($v$: node, $i$: integer) \{\\
\* $left \gets  i \cdot D$;\\
\* $right \gets  \min\{(i+1)\cdot D-1, \deg(v)\}$;\\
\* $deg^* \gets right - left + 1$;\\
\* \textbf{if} $left >0$ \textbf{then} $deg^* \gets deg^* +1$;\\
\* \textbf{if} $right < \deg(v)$ \textbf{then} $deg^* \gets deg^* +1$;\\
\* \textbf{return} $deg^*$;\\
\}\\
\\
// \emph{Return a port at node $(v,i)$ in $G^*$ chosen uniformly at random}\\
\textbf{function} get\_random\_port* ($v$: node, $i$: integer) \{\\
\* $left \gets  i \cdot D$;\\
\* $right \gets  \min\{(i+1)\cdot D-1, \deg(v)\}$;\\
\* $deg^* \gets$ get\_degree* $(v,i)$;\\
\* \textbf{with probability} $(right - left + 1) / deg^*$ \textbf{do}\\
\* \* \textbf{return} integer from range $[left, right]$ chosen uniformly at random;\\
\* $neighbors \gets \emptyset$;\\
\* \textbf{if} $left >0$ \textbf{then} $neighbors \gets neighbors \cup \{`prev'\}$;\\
\* \textbf{if} $right < \deg(v)$ \textbf{then} $neighbors \gets neighbors \cup \{`next'\}$;\\
\* \textbf{return} element of $neighbors$ chosen uniformly at random;\\
\}\\

\newpage
\section*{Appendix B: Auxiliary claims}

\begin{lemma}\label{lemrev}
For all $i,j\in V$, $\E_i N_j(t) = \E_j N_i(t).$
\end{lemma}

\begin{proof}
Let $X_u(\tau)$, $u\in\{i,j\}$, denote the random variable equal to $1$ if a walk of length $\tau$ is located at $u$ after $\tau$ steps, and $0$ otherwise. Since $RW(G_1)$ is a reversible Markovian process, by the properties of the $\tau$-th power of the transition matrix of the walk (cf.~\cite{AF}, Chapter 3.1), we have:
$$\pi(i) \E_i X_j(\tau) = \pi(j) \E_j X_i(\tau).$$
Since $\pi(i) = \pi(j) = \frac{1}{n}$, it follows that $\E_i X_j(\tau) = \E_j X_i(\tau)$. Taking into account that $N_u(t) = \sum_{\tau=0}^{t-1} X_u(\tau)$, $u\in\{i,j\}$, by linearity of expectation we obtain the claim.
\end{proof}

\begin{lemma}\label{lembla}
For any node $i\in V$:
\begin{equation}\label{eqU}
\E_\pi N_i(t) = \frac{t}{n}
\end{equation}
and for any arc $e_{ij}$ of $G$ corresponding to an edge $\{i,j\}\in E$:
\begin{equation}\label{eq6}
\E_\pi N_{e_{ij}}(t) = \frac{t}{n}\min\left\{\frac{1}{\deg (i)}, \frac{1}{\deg(j)}\right\}.
\end{equation}
\end{lemma}
\begin{proof}
Follows directly from the stationary distribution of the Metropolis-Hastings walk on nodes and edges.
\end{proof}

\begin{proposition}\label{lbm}
Any algorithm for USTCON requires time $\Omega(m)$.
\end{proposition}
\begin{proof}
Consider a generic instance of USTCON defined as follows. Take two disjoint copies of some $2$-edge-connected graph $H$ on $n/2$ nodes, with one distinguished node $x$. The two copies of $H$ are assigned the subscripts $1$ and $2$, respectively. Now, as the considered instance of USTCON we use, with probability $1/2$, the disconnected graph $H_1 \cup H_2$ with $\s=x_1$ and $\t=x_2$. Otherwise, we pick an edge $\{u,v\}$ of $H$ uniformly at random, and use as the instance the connected graph $H_1 \cup H_2 \cup \{\{u_2,v_1\},\{u_1,v_2\}\}\setminus \{\{u_1,v_1\},\{u_2,v_2\}\}$, likewise with $\s=x_1$ and $\t=x_2$. Subject to a choice of node identifiers in the representations, the connected and disconnected instances differ on precisely $4$ memory cells in the adjacency lists of the graph (for nodes $u_1$, $u_2$, $v_1$, and $v_2$), and these cells, taken over the choices of edge $\{u,v\}$, form a partition of the memory representation of the graph. Consequently, the expected number of memory reads for an algorithm deciding connectivity with probability $1/2 +p$ is lower-bounded by $p \cdot m/4$, and is thus $\Omega(m)$, within the range $n \leq m \leq n^2/8-O(n)$.
\end{proof}

\section*{Appendix B: Proofs of technical lemmas}


\subsection*{Proof of Lemma~\ref{lemiii}}

The interested reader may see this proof as an analogue of the discussion for short random walks in regular graphs, cf.~Aldous and Fill, Chapter 6, Proposition 16.

\emph{Claim} $(i)$: Consider a shortest path $P$ in graph $G$ from $i$ to a nearest vertex $j \in V\setminus A$. Let $P = (i_0, i_1, \ldots, i_{a}, j)$, where $i_0=i$, and $i_l \in A$,  for $0\leq l \leq a$. Let $G^\circ$ be the subgraph of $G$ induced by nodes from set $A$, their neighbors in $G$, and node $j$: $G^\circ = G [A \cup N(A) \cup \{j\}]$. Since any random walk in $G$ which starts from $i$ and does not enter $V\setminus A$ is confined to nodes and edges of graph $G^\circ$, we have the following relation between the walks $RW(G_1)$ and $RW(G^\circ_1)$:
\begin{equation}\label{eq1}
\E_i T_{V\setminus A} \leq \E_i T^\circ_{V\setminus A} = \E_i T^\circ_j < Com^\circ(i,j) = R^\circ(i,j)\sum_{e \in \vec E(G^\circ)} w(e),
\end{equation}
where the latter equality follows from the electrical network representation of random walks. The resistance $R^\circ(i,j)$ is upper-bounded by the resistance of the series connection going through the nodes of path $P$ in $G$:
$$
R^\circ(i,j) \leq \tfrac{1}{w(e_{i_0 i_1})} + \tfrac{1}{w(e_{i_1 i_2})} + \ldots + \tfrac{1}{w(e_{i_{a-1} i_a})} + \tfrac{1}{w(e_{i_{a} j})} =
$$
$$
= \max\{\deg(i_0),\deg(i_1)\} + \ldots + \max\{\deg(i_{a-1}),\deg(i_a)\} + \max\{\deg(i_{a}),\deg(i_j)\} < 2\sum_{l=0}^{a-1} \deg(i_l) + 2\Delta.
$$
Since the path $P_s = (i_0,i_1,\ldots,i_{a-1})$ is a shortest path in graph $G$ between nodes $i_0$ and $i_{a-1}$, such that $P_s \subseteq A$ and $\Gamma(P_s) \subseteq A$, it follows that (cf.~\cite{AF}):
$$
\sum_{l=0}^{a-1} \deg(i_l) \leq 3 |A|,
$$
and:
\begin{equation}\label{eq2}
R^\circ(i,j) < 6|A| + 2\Delta.
\end{equation}
Since the total weight of edges and self-loops of $G$ incident to a vertex in $V$ is equal to $1$, we have:
\begin{equation}\label{eq3}
\sum_{e \in \vec E(G^\circ)} w(e) \leq \sum_{v\in A} \left( \sum_{u \in \Gamma(v) \cup \{v\}} w(e_{vu}) \right) +  \sum_{u \in \Gamma(j) \cup \{j\}} w(e_{ju}) \leq |A| + 1.
\end{equation}
Claim $(i)$ follows from inequalities~\eqref{eq1},~\eqref{eq2}, and~\eqref{eq3}.

\medskip
\noindent
\emph{Claim} $(ii)$: Suppose that $s = \sqrt {6t} \geq \frac t n$, and let:
$$
A = \{j\in V : \E_i N_j(t) > s\}.
$$
Since the considered walk hits nodes from $V$ a total of (at most) $t$ times, we have $|A| < \frac{t}{s} \leq n$, and the considerations performed in the proof of Lemma~\ref{lemiii}$(i)$ can be applied for the above-defined set $A$.

First, we bound the expected number of returns to node $i$ for a walk starting at $i$ before reaching $V\setminus A$ for the first time:
$$
\E_i N_i (T_{V\setminus A}) = 1 + (1-{\Pr_i[T_{V\setminus A}<T_i]})\cdot \E_i N_i (T_{V\setminus A}) \implies \E_i N_i (T_{V\setminus A}) = \frac{1}{\Pr_i[T_{V\setminus A}<T_i]}.
$$
Taking into account~\cite{AF} (Chapter 3, eq. (28) and Corollary 11) and bound~\eqref{eq2}, we have:
$$
\E_i N_i (T_{V\setminus A}) = \frac{1}{\Pr_i[T_{V\setminus A} < T_i]} = \pi(i) \cdot R(i,j)\cdot\!\!\sum_{e \in \vec E(G)} w(e) \;\leq\; \pi(i) \cdot R^\circ(i,j)\cdot\!\!\sum_{e \in \vec E(G)} w(e) <
$$
\begin{equation}
\label{eq4}
< \frac{1}{n}\cdot (6|A|+2\Delta) n = 6|A|+2\Delta < 6 \frac{t}{s} + 2\Delta.
\end{equation}
It follows from Lemma~\ref{lemrev} that the definition of set $A$ may be rewritten as:
$$
A = \{j\in V : \E_j N_i(t) > s\}
$$
Thus, $V\setminus A = \{j\in V : \E_j N_i(t) \leq s\}$, which means that if a walk starting from $i$ reaches $V\setminus A$, it will return to $i$ at most $s$ times in expectation before time $t$. So, using~\eqref{eq4}, we obtain the claim:
$$
\E_i N_i(t) \leq \E_i N_i(T_{V\setminus A}) + s < 6\frac{t}{s} + s + 2\Delta = 2\sqrt{6t} + 2\Delta < 5\sqrt t + 2\Delta.
$$

\qed


\subsection*{Proof of Lemma~\ref{lemA}}

Fix an arbitrary arc $e_{ij}$, with $\{i,j\}\in E$. We will bound the sought probability from the inequality:
\begin{equation}\label{eq5}
\E_{\pi} N_{e_{ij}} (t) \leq \Pr_{\pi} [T_{e_{ij}} < t]\ \E_{e_{ij}} N_{e_{ij}} (t)
\implies
\Pr_{\pi} [T_{e_{ij}} < t] \geq \frac{\E_{\pi} N_{e_{ij}} (t)}{\E_{e_{ij}} N_{e_{ij}} (t)}.
\end{equation}
The expected number of traversals of $e_{ij}$ for a walk of even length starting from the stationary distribution on $V$ is given by equation~\eqref{eq6}.

To bound the expectation from the denominator of~\eqref{eq5}, we note that by Lemma~\ref{lemiii}$(ii)$,  $\E_i N_i (t)  < 5 \sqrt t + 2 \Delta$, and that arc $e_{ij}$ is chosen with probability $\min\left\{\frac{1}{\deg (i)}, \frac{1}{\deg(j)}\right\}$ during each visit to $i$:
$$
\E_i N_{e_{ij}} (t)  < (5 \sqrt t + 2\Delta) \min\left\{\frac{1}{\deg (i)}, \frac{1}{\deg(j)}\right\}.
$$
Considering a walk starting from a traversal of arc $e_{ij}$, we observe that after its traversal of  $e_{ij}$ the walk must return to node $i$ before traversing $e_{ij}$ again:
$$
\E_{e_{ij}} N_{e_{ij}} (t)  \leq  1 + \E_{j} N_{e_{ij}} (t) < 1 + \E_i N_{e_{ij}} (t) < 1 + (5 \sqrt t + 2 \Delta)  \min\left\{\frac{1}{\deg (i)}, \frac{1}{\deg(j)}\right\} \leq
$$
\begin{equation}\label{eq7}
\leq (5 \sqrt t + 3 \Delta) \min\left\{\frac{1}{\deg (i)}, \frac{1}{\deg(j)}\right\}.
\end{equation}
By combining inequalities~\eqref{eq6},~\eqref{eq5},~\eqref{eq7}, and taking into account that $t > \Delta^2$, we obtain the claim:
$$
\Pr_{\pi} [T_{ e_{ij}} < t] > \frac{t}{n (5\sqrt t + 3 \Delta)} > \frac{t}{8n \sqrt t} > 0.1 \frac{\sqrt{t}}{n}.
$$
\qed

\subsection*{Proof of Lemma~\ref{lemB}}

Pick a node $j \in V$ according to the uniform probability distribution $\pi$. We will bound the sought probability from the inequality:
\begin{equation}\label{eq5b}
\E_{i} N_{j} (t) \leq \Pr_{i} [T_{j} < t]\ \E_{j} N_{j} (t)
\implies
\Pr_{i} [T_{j} < t] \geq \frac{\E_{i} N_{j} (t)}{\E_{j} N_{j} (t)}.
\end{equation}
Taking into account Lemma~\ref{lemrev} and condition~\eqref{eqU}, and noting that $j$ is chosen according to the uniform distribution $\pi$ on $V$, we have:
\begin{equation}\label{eq6a}
\E_{i} N_{j} (t) = \E_{j} N_{i} (t) = \E_{\pi} N_{i} (t) = \frac{t}{n}.
\end{equation}
The expectation from the denominator of~\eqref{eq5b} is bounded by Lemma~\ref{lemiii}$(ii)$, $\E_{j} N_{j} (t) < 5\sqrt t + 2\Delta$. By combining the above relations, and taking into account that $t > \Delta^2$, we obtain:
$$
\Pr_{i} [T_{j} < t] > \frac{t}{n (	5\sqrt t + 2 \Delta)} > \frac{t}{7n \sqrt t} > 0.1 \frac{\sqrt{t}}{n}.
$$
\qed

\subsection*{Proof of Lemma~\ref{lemGood}}

Fixing a connected component $H \subseteq G$ with $n_H \geq n_p/6$, we introduce the following notation for a set of landmarks $L$:
\begin{itemize}
\item let $L_H = L\cap V(H)$,
\item let $X(L)$ denote the event that $|L_H| \geq \frac{1}{2} p \frac{n_H}{n}$,
\item let $F_1(L)$ be the random variable over $L$ describing the maximum, over all non-loops arcs $e$ belonging to $H$, of the probability that a set of $p$ random walks $RW(G_1)$ of length $\tau = n_p^2$ each, with one random walk originating from each landmark from $L$, does not cover arc $e$.
\item let $F_2(L)$ be the random variable over $L$ describing the maximum, over all nodes $u \in V(H)$, of the probability that a random walk $RW(G_1)$ of length $\tau = n_p^2$, originating from $u$, does not hit any landmark of $L$.
\end{itemize}
Suppose that $L$ is a set of $p$ nodes picked according to the uniform distribution $\pi^p$ on $V^p$. To prove the claim of the Lemma, we need to show the following bound:
\begin{equation}\label{eq11}
\Pr_{L\sim \pi^p} [F_1 > n^{-1} \wedge  F_2 > n^{-1}] <\frac{1}{2n}.
\end{equation}
We observe that each landmark from $L$ belongs to $V(H)$ with probability $n_H/n$. Let $L_H = L\cap V(H)$. A w.h.p.\ lower bound on the size of $L_H$ follows from the Chernoff bound applied to $p$ Bernoulli trials with success probability $n_H/n$:
\begin{equation}\label{eq12}
\Pr_{L\sim \pi^p} [ X ] \ \geq 1 - e^{-\frac{1}{8}p\frac{n_H}{n}} \geq 1 - e^{-\frac{1}{48}p\frac{n_p}{n}}\geq 1 - e^{-\frac{\gamma}{48}\log n} > 1 - \frac{1}{4n},
\end{equation}
where we took into account that $n_p \geq \gamma \frac{n}{p} \log n$, and that $\gamma=60 > 48$. In the following, we only need to show that, conditioned on the event $X(L)$ holding, $L$ is a good set of landmarks with probability $1 - \frac{1}{4n}$. Note that all the landmarks from $L_H$ are distributed uniformly at random on $V(H)$, also when conditioned on $X(L)$.

To bound $F_1(L)$, fix a non-loop arc $e$ of $H$ as the arc maximizing the failure probability in the definition of $F_1(L)$. By applying Lemma~\ref{lemA} to graph $H$, the probability that a walk $RW(H_1)$ of length $\tau=n_p^2$, starting from the uniform distribution on $V(H)$, does not cover arc $e$, is at most $1 - \frac{0.1 n_p}{n_H}$. Thus, considering that:
$$
|L_H| \geq \frac{1}{2} p \frac{n_H}{n} = \frac{n_H \cdot 3 \log n}{6 \frac n p\log n} \geq \frac{n_H \cdot 3 \log n}{0.1 n_p},
$$
the probability $F_{1,e}(L)$ that no random walk starting from a landmark hits arc $e$ is bounded by:
$$
\E_{L\sim \pi^p} \left[F_{1} \ \big|\ X \right] < \left(   1 - \frac{0.1 n_p}{n_H}  \right)^{\frac{n_H}{0.1 n_p} 3 \log n} < 2^{-3 \log n} < n^{-3}.
$$

Likewise, to bound $F_2(L)$, fix a node $u \in V(H)$ maximizing the probability that a walk $RW(G_1)$ of length $\tau = n_p^2$, originating from $u$, does not hit any landmark of $L$. By Lemma~\ref{lemB}, the probability that the considered walk of length $\tau$ does not cover a node chosen according to the uniform distribution on $V(H)$, is at most $1 - \frac{0.1 \sqrt{\tau}}{n_H}$. Thus, taking into account that $|L_H| > \frac{n_H}{0.1 \sqrt{\tau}} 3 \log n$, the probability that the walk does not hit any landmark can once again be bounded as less than $n^{-3}$:
$$
\E_{L\sim \pi^p} \left[F_2\ \big|\ X \right] < n^{-3}.
$$
It follows that:
$$
\E_{L\sim \pi^p} \left[F_1 + F_2\ \big|\ X \right] < 2n^{-3},
$$
and by the Markov bound:
\begin{equation}\label{eq13}
\Pr_{L\sim \pi^p} [F_1 + F_2> n^{-1}\  \big|\ X  ] < \frac{2}{n^2} < \frac{1}{4n}.
\end{equation}
Now, inequalities~\eqref{eq12} and~\eqref{eq13} imply that inequality~\eqref{eq11} holds, which completes the proof.
\qed

\subsection*{Proof of Proposition~\ref{probla}}

We begin by observing that the unbiased random walk on $G$ can be described as a weighted Metropolis-Hastings walk $RW(G_{f_c})$, where, for all $u\in V$, the potential function on nodes is given as $f_c(u) = c\deg(u)$, where $c>0$ is an arbitrarily chosen constant of proportionality ($w(e) = c$ for all edges). Now, looking at the electrical networks analogy, by identifying with each other the corresponding nodes of the electrical networks describing the walks $RW(G_{f_c})$ and $RW(G_1)$, and leaving the edges of both these networks in parallel connection, we obtain a new network on $G$ with edge weights $w_f$ given by:
$$
w_f(e) = w_{f_c}(e) + w(e),
$$
corresponding to the potential function on nodes:
$$
f(u) = f_c(u) + 1 = c\deg(u) + 1.
$$
It follows that the resistance of replacement of the network of $RW(G_f)$ for any two nodes $u,v\in V$ can be bounded as:
$$
R_{G_f}(u,v) \leq R_{G_{f_c}}(u,v) \quad \text{and} \quad R_{G_f}(u,v) \leq R_{G_1}(u,v).
$$
Moreover, the following relations hold between resistances and commute times:
$$
Com_{G_{f_c}}(u,v) = 2cm R_{G_{f_c}}(u,v)
$$
$$
Com_{G_1}(u,v) = n R_{G_1}(u,v)
$$
$$
Com_{G_{f}}(u,v) = (2cm + n) R_{G_f}(u,v)
$$
Fixing $c = \frac{1}{d} = \frac{n}{2m}$, i.e., $2cm = n$, we obtain from all of the above relations:
\begin{equation}\label{eqXA}
Com_{G_f}(u,v) = O(\min\{Com_{G_{f_c}}(u,v),Com_{G_1}(u,v)\}).
\end{equation}
\qed
\end{document}